\documentclass{ifacconf}

\usepackage{graphicx,color} 
\usepackage{tikz}
\usepackage{natbib} 
\usepackage{algorithmicx}
\usepackage{algorithm}
\usepackage{algpseudocode}
\newtheorem{assumption}{Assumption}

\newtheorem{example}{Example}
\newtheorem{lemma}{Lemma}

\usepackage{float}
 
\usepackage{enumitem}
\setlist[itemize]{leftmargin=*}
\setlist[enumerate]{leftmargin=*}
\usepackage{comment}
\usepackage{amsmath}
\usepackage{amssymb}
\usepackage{mathtools}

\DeclareMathOperator{\diag}{diag}

\DeclareMathOperator{\Img}{img}
\DeclareMathOperator{\row}{row}

\def\S{\mathcal{S}}

\def\M{\mathcal{M}}
\def\N{\mathcal{N}}

\renewcommand{\O}{\mathcal{O}}
\newcommand{\RR}{\mathbb{R}}
\newcommand{\NN}{\mathbb{N}}
\usepackage{bm}

\usepackage{array}
\usepackage{amsmath}
\usepackage{alphalph}

\makeatletter
\newcounter{mysubequation}
\newenvironment{mysubequations}{%
  \refstepcounter{equation}%
  \protected@edef\theparentequation{\theequation}%
  \setcounter{mysubequation}{0}%
  \setcounter{equation}{0}%
  \def\theequation{\theparentequation\alphalph{\value{equation}}}%
  \ignorespaces
}{%
  \setcounter{equation}{\value{mysubequation}}%
  \ignorespacesafterend
}
\makeatother

\begin{document}
\begin{frontmatter}

\title{Interval Observer Design Using Observability Decomposition\\for Detectable Linear Systems}

\author[First]{G. Q. Bao Tran}
\author[Second]{Thach Ngoc Dinh}
\author[Third]{Zhenhua Wang}

\address[First]{Coordinated Science Laboratory, University of Illinois Urbana-Champaign, Urbana, IL 61801, USA\\(e-mail: baotran@illinois.edu).}
\address[Second]{Cedric-Lab, Conservatoire National des Arts et M\'etiers (CNAM), Paris, France (e-mail: ngoc-thach.dinh@lecnam.net).}
\address[Third]{School of Astronautics, Harbin Institute of Technology, Harbin, China (e-mail: zhenhua.wang@hit.edu.cn).}

\begin{abstract}
We provide a systematic interval observer design method for detectable linear time-invariant (LTI) systems, where a part of the state is observable from the measured output. An observability-based invertible LTI transformation decomposes the state into two parts. The first part is decoupled from the other and observable from the output, while the second is affected by the first, does not appear in the output, but is detectable. A Sylvester-based LTI interval observer is designed for the first part. For the second part, a Jordan-based linear time-varying interval observer is built, treating the interaction from the first part as inputs with known bounds. The intervals in the original coordinates are constructed either by inverting the decomposition online for the intervals in the transformed coordinates or by directly implementing the observer written in the original coordinates. Academic examples illustrate the interest of our approach.
\end{abstract}

\end{frontmatter}

\section{INTRODUCTION}\label{sec_intro}
Interval observers, unlike asymptotic observers, provide upper and lower bounds on the actual state of a perturbed system at each time instant, possibly from noisy measurements. Since their introduction in the pioneering work~\citep{Gouze-00}, interval observers have evolved in various directions, driven by the critical role of state estimation in monitoring~\citep{DiIto:2017:SCIEJournal}, fault detection~\citep{Wangetal:2024:SCIS}, and control applications~\citep{mex2,dinh2026set}. Although this design approach has proven effective, it relies on specific assumptions. Specifically, interval observers can be constructed when both the initial conditions and the time-varying uncertainties are bounded by known vectors, and their feasibility relies on a direct or indirect notion of a Metzler (cooperative) system in the case of continuous time (CT), or a non-negative system in discrete time (DT).

The seminal works~\citep{Mazenccontinu} in CT and~\citep{mazenc2014interval} in DT have proposed time-varying changes of coordinates that leverage the Jordan form of the real-valued dynamics matrix of a linear time-invariant (LTI) system, transforming it into the proper forms necessary for interval observer design. However, these transformations inherently lead to \emph{time-varying} observers, even for LTI systems. In our recent work~\citep{baoThachAut}, we have proposed Sylvester-based LTI designs, which have a built-in transformation into the same dimensions obtained from solving offline a Sylvester equation. This method, inspired by Luenberger's observer~\citep{luenberger}, also extends to linear time-varying systems in the mentioned work. However, this transformation, even in the LTI case, is shown to be invertible under the \emph{observability} of the system (i.e., the output and its derivatives uniquely determine the state). 
Hence, it is inapplicable when the system is only \emph{detectable} (i.e., solutions producing the same output converge asymptotically to each other). 
In the literature,~\citep{Raissietal,gu} also employ Sylvester equations for interval observer design for LTI systems, while~\citep{Efimov-12,cacace} succeed in building these estimators under less straightforward assumptions.

This paper proposes an approach that effectively fuses the aforementioned methods for detectable LTI CT and DT systems, where a part of the state is moreover observable from the output. As mentioned above, the Jordan-based approach requires only detectability but results in a time-varying observer, while the Sylvester-based design yields an LTI observer but requires observability, and we herein seek to combine them strategically, leaving aside the question of performance (i.e., tightness of the interval bounds) for future studies. We first perform an observability decomposition that separates the observable part of the state from the one that is only detectable. The dynamics of the first part and the measurement are independent of the second part. With this decoupling, we can apply the two mentioned designs to the respective parts, treating the interaction of the first part on the second one as a disturbance with bounds given by the first interval observer. We recover the bounds in the original coordinates either by inverting the decomposition online for the intervals or by directly implementing the observer in the given coordinates. Our advantage compared to classical methods is that this combination of interval observers maximizes the use of observability to make as much of the design time-invariant as possible, while resorting to more intricate transformations only under detectability. As a result, only part of the observer remains time-varying, while the rest becomes constant and thus can be computed offline, potentially considerably reducing computational complexity in high-dimensional systems. Our approach is illustrated step-by-step through academic examples.

\emph{Notations:} Inequalities like $a \leq b$ for real vectors $a$, $b$ or $A \leq B$ for real matrices $A$, $B$ are component-wise. For a matrix $M \in \RR^{n \times m}$ with entries $m_{ij}$, define $M^\oplus$ as the matrix in $\RR^{n \times m}$ whose entries are $\max\left\{0, m_{ij}\right\}$ and let $M^\ominus = M^\oplus - M$. 
We present results in CT and DT together. Let $x_t$ denote the current state, combining $x(t)$ in CT and $x_k$ in DT. Similarly, $x^+_t$ denotes its time derivative$\slash$jump, combining $\dot{x}(t)$ in CT and $x_{k+1}$ in DT. 

\begin{lemma}\citep[Section II.A]{Efimov-12}
\label{lem_pm}
Consider vectors $a$, $\overline{a}$, $\underline{a}$ in $\mathbb{R}^{n}$ such that $\underline{a} \leq a \leq \overline{a}$. For any $A \in \RR^{m\times n}$, we have $A^{\oplus}\underline{a} -A^{\ominus}\overline{a} \leq Aa \leq A^{\oplus}\overline{a} -A^{\ominus}\underline{a}$.
\end{lemma}
   
\begin{lemma}\label{lem_Pc}\citep[Theorem $2$]{Mazenccontinu}
For any Hurwitz constant matrix $F \in \RR^{n\times n}$, there exist a time-varying invertible real matrix $t \mapsto P_t\in \RR^{n\times n}$, a scalar $\sigma>0$, and a Metzler Hurwitz constant matrix $\Lambda\in \RR^{n\times n}$ such that for all $t \geq 0$, $\|P_t\|+\|(P_t)^{-1}\| \leq \sigma$ and $\Lambda = (P_t^+ + P_t F)(P_t)^{-1}$. 
\end{lemma}

\begin{lemma}\label{lem_Pd}\citep[Theorem $4$]{mazenc2014interval}
For any Schur constant matrix $F\in \RR^{n\times n}$, there exist a sequence of invertible real matrices $(P_t)_{t \in \NN}$ in $\RR^{n\times n}$, a scalar $\sigma>0$, and a non-negative Schur constant matrix $\Lambda \in \RR^{n\times n}$ such that for all $t \in \NN$, $\|P_t\|+\|(P_t)^{-1}\| \leq \sigma$ and $\Lambda = P_t^+F(P_t)^{-1}$. 
\end{lemma}

\section{OBSERVABILITY DECOMPOSITION OF PERTURBED DETECTABLE LTI SYSTEMS}
Consider a perturbed CT$\slash$DT LTI system
\begin{equation}\label{eq:sys}
    x_t^+=Fx_t + u_t + Dd_t, \qquad
        y_t=Hx_t + Ww_t,
\end{equation}
with state $x_t\in\RR^{n_x}$, known input $u_t \in \RR^{n_x}$, measured output $y_t \in \RR^{n_y}$, and unknown disturbance and noise $(d_t,w_t)$ in $\RR^{n_d} \times \RR^{n_w}$. We design for system~\eqref{eq:sys} an \emph{interval observer} as defined in~\citep[Def. $1$]{Mazenccontinu} for CT and~\citep[Def. $1$]{mazenc2014interval} for DT.

If the pair $(F,H)$ is observable, then we systematically use the Sylvester-based design in~\citep{baoThachAut}, which involves simply solving a Sylvester equation for a constant matrix $T$, such that $Tx_t$ follows linear, stable, and Metzler$\slash$non-negative dynamics, for which we can always build an interval observer, and inverting this linear transformation to recover the $x$-coordinates bounds. However, in this work, we tackle the case where only a part of the state $x_t$ is observable and the rest is detectable. 
The lack of observability of the full state renders the method in~\citep{baoThachAut} inapplicable since $T$ would then be non-invertible. We could have come back to using Jordan form-based transformations for the whole $F$ as in~\citep{Mazenccontinu} for CT and~\citep{mazenc2014interval} for DT, which requires only detectability, but here we choose another path based on decomposition as follows. 

Let the observability matrix of system~\eqref{eq:sys} be 
\begin{equation}
\O := \row(H,H F, \ldots, HF^{n_x-1}),
\end{equation}
and assume it is of non-zero rank $n_{o}:=\dim \Img \O< {n_x}$. This means that only a non-empty part of $x_t$ of dimension $n_o < n_x$ is observable from $y_t$. The remaining part, assumed to be (asymptotically) detectable, has dimension $n_{no}$, so that $n_o + n_{no} = n_x$. Consider a basis $(v_i)_{1 \leq i \leq n_x}$ of $\RR^{n_x}$ such that $(v_i)_{1 \leq i \leq n_o}$ is a basis of the observable subspace $\Img \O$ and $(v_i)_{n_o+1 \leq i \leq n_x}$ is a basis of the non-observable subspace $\ker \O$. Then, we define the invertible matrix $\M :=
\begin{pmatrix}
	\M_o & \M_{no}
\end{pmatrix} \in \RR^{n_x \times n_x}$ where\begin{subequations}
\begin{align}
	\M_o &:=
	\begin{pmatrix}
		v_1 & \ldots & v_{n_o}
	\end{pmatrix} \in \RR^{n_x \times n_o}, \\
	\M_{no} &:=
	\begin{pmatrix}
		v_{n_o+1} & \ldots & v_{n_x}
	\end{pmatrix} \in \RR^{n_x \times n_{no}},
\end{align}
\end{subequations}
which satisfies, by definition,
\begin{equation}
	\label{eq:prop_obszo}
	\O \M_{no} = 0, \qquad H \M_{no} = 0.
\end{equation}
We denote $\N := \M^{-1}$ which we decompose consistently into two parts
$\N =: \begin{pmatrix} \N_o \\ \N_{no}
\end{pmatrix}$,
so that $\N_o x_t$ represents the part of the state that is (differentially) observable from $y_t$. The change of coordinates
\begin{equation}\label{eq:trans}
	x_t \mapsto z_t = \begin{pmatrix}z_{o,t}\\ z_{no,t} \end{pmatrix} = \N x_t = \begin{pmatrix}\N_o  \\ \N_{no} \end{pmatrix}x_t,
\end{equation}
with the inverse transformation given by
\begin{equation}\label{eq:inv}
	x_t = \M z_t =\M_o z_{o,t} + \M_{no}z_{no,t},
\end{equation}
transforms system~\eqref{eq:sys}
into
\begin{subequations}\label{eq:sysz}
\begin{equation}
		\left\{
		\begin{array}{@{}r@{\;}c@{\;}l@{}}
			z_{o,t}^+ &=& F_o z_{o,t} + \N_o u_t + D_o d_t\\ z_{no,t}^+ &=& F_{noo}z_{o,t} + F_{no}z_{no,t} + \N_{no}u_t + D_{no}d_t
		\end{array}
		\right.
        \end{equation}
        with the output
        \begin{equation}
			y_t = H_oz_{o,t} + W w_t,
	\end{equation}
\end{subequations}
    where $F_o = \N_o F \M_o \in \RR^{n_o \times n_o}$, $D_o = \N_o D \in \RR^{n_o \times n_d}$, $F_{noo} = \N_{no} F \M_o \in \RR^{n_{no} \times n_o}$, $F_{no} = \N_{no} F \M_{no} \in \RR^{n_{no} \times n_{no}}$, $D_{no} = \N_{no}D \in \RR^{n_{no} \times n_d}$, and $H_o = H\M_o \in \RR^{n_y \times n_o}$. The form~\eqref{eq:sysz}, obtained from observability decomposition, applies to both CT and DT LTI systems. 
    Lemma~\ref{lem_decompose} establishes the properties of system~\eqref{eq:sysz}.    
\begin{lemma}\label{lem_decompose}
The following properties hold:
\begin{itemize}[leftmargin=*,nosep]
\item $\N_o F \M_{no} = 0$ and $H\M_{no} = 0$;
    \item The pair $(F_o,H_o)$ is observable;
    \item The matrix $F_{no}$ is Hurwitz (in CT) or Schur (in DT).
\end{itemize}
\end{lemma}

\begin{pf}
    The first two items are the properties of the projection of $(F,H)$ onto the observable subspace---see~\citep[Lemma $4.2$]{antsaklis1997linear}. The third item follows from the definition of asymptotic detectability of LTI systems: non-observable state components (in our case $z_{no,t}$) must be asymptotically stable. \hfill $\blacksquare$
\end{pf}

The following academic example, re-visited throughout this paper, illustrates the observability decomposition.
\begin{example}\label{eg1}
    Consider a DT system of form~\eqref{eq:sys} with $F = \left(\begin{smallmatrix}
       -1 &0 &1 &0\\ 1 &-0.5& -1 &1\\ 0 &0 &0& -1\\ 0& 0 &-1 &0
    \end{smallmatrix}\right)$ and $H = \begin{pmatrix}
        1 & 0 & 1 & 1
    \end{pmatrix}$. Consider known input $u_t = (\sin(t),0,-0.5\cos(t),0)$, input disturbance $d_t = 0.02\cos(5t)$ with gain $D = (1,0,0,0)$, and measurement noise $w_t = 0.01\sin(20t)$ with gain $W = 1$. We would like to build an interval observer for this system, assuming some bounds on the disturbance and noise. The observability matrix $\O = \left(\begin{smallmatrix}
        H \\ HF \\ HF^2\\  HF^3
    \end{smallmatrix}\right)= \left(\begin{smallmatrix}
        1 & 0 & 1 & 1 \\ -1 & 0 & 0 & -1 \\ 1 & 0 & 0 &0  \\ -1 & 0 & 1 & 0
   \end{smallmatrix}\right)$ is of rank $3$ (the second state is only detectable), so the method in~\citep{baoThachAut} is inapplicable. Indeed, when solving $TF = AT + BH$ for some $A$ Schur with different eigenvalues from $F$, the obtained $T$ is non-invertible. We could have used the Jordan-based transformation in~\citep{mazenc2014interval}, which requires computing a $4\times4$ matrix and its inversion online at each iteration, followed by incorporating it into the observer. Following this work, we now perform an observability decomposition for this system, with $\M_o = \left(\begin{smallmatrix}
            1 & 0 & 0 \\0& 0 & 0\\0 & 1 & 0\\0 & 0 & 1
    \end{smallmatrix}\right)$ and $\M_{no} = \left(\begin{smallmatrix}
            0 \\1 \\0\\0 
   \end{smallmatrix}\right)$. We get $F_o = \left(\begin{smallmatrix}
        -1  &   1 &   0\\
     0  &   0  &  -1\\
     0  &  -1 &    0
    \end{smallmatrix}\right)$ and $H_o = \begin{pmatrix}
        1 & 1 & 1
    \end{pmatrix}$, which form an observable pair. We also get $F_{no} = -0.5$, which is Schur since $x_2$ is detectable.
\end{example}


\section{DECOMPOSITION-BASED\\INTERVAL OBSERVER DESIGN}
In this section, we design a decomposition-based interval observer for system~\eqref{eq:sys}. We start by making the following standard assumptions about the bounds of the initial condition and uncertainties.

\begin{assumption}\label{ass_sys}
    For system~\eqref{eq:sys}, we assume that:
\begin{itemize}[leftmargin=*,nosep]
    \item There exist $(\overline x_0,\underline x_0) \in \RR^{n_x} \times \RR^{n_x}$ such that $\underline x_0 \leq x_0 \leq \overline x_0$;
    \item There exist known finite vectors $(\overline d_t,\underline d_t,\overline w_t,\underline w_t) \in \RR^{n_d} \times \RR^{n_d} \times \RR^{n_w} \times \RR^{n_w}$ such that $\underline d_t \leq d_t \leq \overline d_t$ and $\underline w_t \leq w_t \leq \overline w_t$ for all $t \geq 0$ (in CT) or $t \in \NN$ (in DT).
\end{itemize}
\end{assumption}

Our decomposition-based interval observer is illustrated in Figure~\ref{fig:obs}. It consists of two designs based on existing methods, one for $z_{o,t}$, another for $z_{no,t}$, followed by an inversion of the decomposition that preserves the interval and convergence properties. In the next sections, we present the respective designs for each part of the state. 
\begin{figure*}[ht]
    \centering
\begin{tikzpicture}[scale=1.75, every node/.style={align=center}]
\node[draw, very thick, text width=3cm, minimum height=1cm, fill=orange!20] (System) at (0.7, 1) {System with state $x_t$};
\node[draw, very thick, text width=3cm, minimum height=1cm, fill=red!20] (Observer1) at (4, 1) {Observer for $z_{o,t}$};
\node[draw, very thick, text width=3cm, minimum height=1cm, fill=blue!20] (Observer2) at (4, -0.2) {Observer for $z_{no,t}$};
\node[draw, very thick, text width=2cm, minimum height=1cm, fill=green!20] (Inversion) at (6.5, 0.4) {Inversion};
\draw[very thick, ->] (System.east) -- (Observer1.west) node[midway, above] {$(y_t,w_t)$};
\draw[very thick, ->] (Observer1.south) -- (Observer2.north) node[midway, right] {$(\overline{z}_{o,t}, \underline{z}_{o,t})$};
\draw[very thick, ->] (Observer1.east) -| (Inversion.north) node[midway, above] {$(\overline{z}_{o,t}, \underline{z}_{o,t})$};
\draw[very thick, ->] (Observer2.east) -| (Inversion.south) node[midway, below] {$(\overline{z}_{no,t}, \underline{z}_{no,t})$};
\draw[very thick, ->] (Inversion.east) -- ++(1.3, 0) node[midway, above] {$(\overline{x}_{t}, \underline{x}_{t})$};
\draw[very thick, ->] (-1.5, 1) -- (System.west) node[midway, above] {$d_t$};
\draw[dashed, very thick] (2.9,-0.68) rectangle (7.3,1.65);

\node[align=center] at (3.4, 1.50) {Observer};
\end{tikzpicture}
    \caption{Structure of the decomposition-based interval observer. The known input $u_t$, compensated in the observers, is neglected for clarity.}
    \label{fig:obs}
\end{figure*}
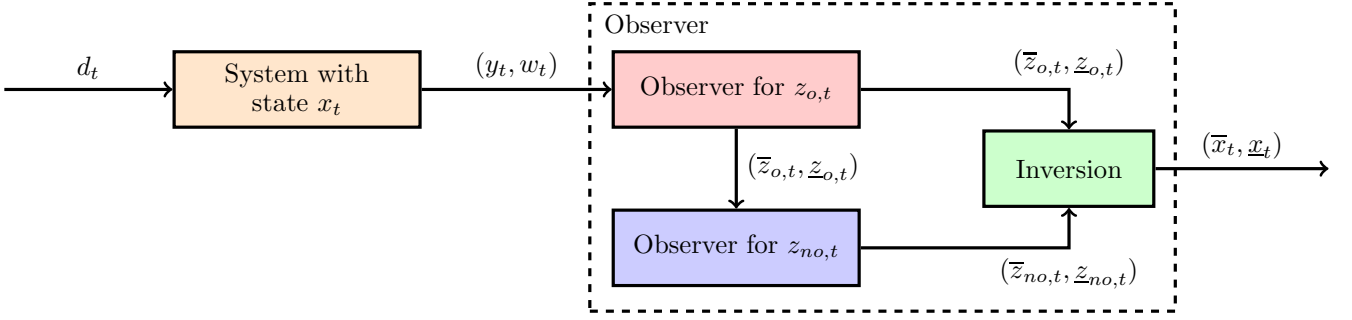

\subsection{Sylvester-Based Design for $z_{o,t}$}
We now design an interval observer for the $z_{o,t}$ part of system~\eqref{eq:sysz}. This is based on the dynamics of $z_{o,t}$ and measurement $y_t$, which are extracted as
\begin{equation}\label{eq:syszo}
    z_{o,t}^+ = F_o z_{o,t} + \N_o u_t + D_o d_t, \quad y_t = H_o z_{o,t} + Ww_t.
\end{equation}
Based on our recent work~\citep{baoThachAut}, we consider an LTI transformation $z_{o,t} \mapsto T z_{o,t}$, where $T \in \RR^{n_o \times n_o}$ is solution to the Sylvester equation
\begin{equation}\label{eq:sylvester}
    TF_o = A_oT + B_oH_o,
\end{equation}
with $A_o \in \RR^{n_o \times n_o}$ Metzler and Hurwitz (in CT) or non-negative and Schur (in DT), and some $B_o  \in \RR^{n_o \times n_y}$, which exists uniquely if and only if the eigenvalues of $A_o$ are different from those of $F_o$. To come back to the $z_o$-coordinates, we rely on the invertibility of $T$, which is guaranteed by the following lemma.

\begin{lemma}\label{lem_T}
For almost any $(A_o,B_o) \in \RR^{n_o \times n_o} \times \RR^{n_o \times n_y}$ with $A_o$ Hurwitz (in CT) or Schur (in DT), $T$ solution to~\eqref{eq:sylvester} exists uniquely and is invertible.
\end{lemma}

\begin{pf}
The existence and uniqueness of $T$ are guaranteed when $A_o$ does not share any eigenvalues with $F_o$~\citep{sylvester}, which is true for almost any $A_o$. The proof of invertibility of $T$ uses~\citep{caronrichard,hucheng} and follows from the observability of $(F_o,H_o)$ and the fact that almost any pair of matrices is controllable~\citep{brivadisRemarks}. See~\citep{baoThachAut} for more details. \hfill $\blacksquare$
\end{pf}

The term \emph{almost any} means that the set of pairs $(A_o,B_o)$ making $T$ non-invertible has zero Lebesgue measure with respect to the set of real matrices. Note that the results in Lemma~\ref{lem_T} become \emph{for any} when $n_y = 1$~\citep{luenberger}, with the additional conditions about $A_o$ not sharing any eigenvalue with $F_o$ and $(A_o,B_o)$ being controllable. For interval observer design, we will moreover select $A_o$ Metzler (in CT) or non-negative (in DT), transforming system~\eqref{eq:syszo} into the proper form for interval observer design by an LTI transformation.

Thanks to the dynamics of $z_{o,t}$ and the output $y_t$ being independent of $z_{no,t}$, denoting $T^{-1}$ as the inverse matrix of $T$, we propose a Sylvester-based LTI interval observer for system~\eqref{eq:syszo}, written directly in the $z_o$-coordinates  following~\citep{baoThachAut} with the dynamics
\begin{subequations}\label{eq:obszo}
\begin{align}
        \overline{\hat{z}}_{o,t}^+ &= F_o \overline{\hat{z}}_{o,t} + \N_o u_t + T^{-1}B_o (y_t - H_o \overline{\hat{z}}_{o,t})  \notag\\&\qquad{} + T^{-1} ((TD_o)^\oplus \overline d_t - (TD_o)^\ominus \underline d_t)  \notag\\&\qquad{}+ T^{-1} ((B_oW)^\ominus \overline w_t - (B_oW)^\oplus \underline w_t),\\
       \underline{\hat{z}}^+_{o,t}&= F_o \underline{\hat{z}}_{o,t} + \N_o u_t + T^{-1}B_o (y_t - H_o \underline{\hat{z}}_{o,t}) \notag\\&\qquad{}+ T^{-1}((TD_o)^\oplus \underline d_t - (TD_o)^\ominus \overline d_t) \notag\\&\qquad{} + T^{-1}((B_oW)^\ominus \underline w_t - (B_oW)^\oplus \overline w_t),
    \end{align}
associated with the initial conditions
\begin{align}\label{eq:obszo_init}
        \overline{\hat{z}}_{o,0} &= T^{-1}((T\N_o)^\oplus \overline{x}_0-(T\N_o)^\ominus \underline{x}_0),\\\label{eq:initxlinm}
        \underline{\hat{z}}_{o,0} &= T^{-1}((T\N_o)^\oplus \underline x_0-(T\N_o)^\ominus \overline x_0),
\end{align}
with the bounds after time $0$
\begin{align}\label{eq:obszo_bdz} 
      \overline z_{o,t}&= (T^{-1})^\oplus T\overline{\hat{z}}_{o,t}- (T^{-1})^\ominus T\underline{\hat{z}}_{o,t}, \\\label{eq:initbxlinm}
        \underline z_{o,t}&= (T^{-1})^\oplus T\underline{\hat{z}}_{o,t} - (T^{-1})^\ominus T\overline{\hat{z}}_{o,t},
\end{align}
and the bounds at time $0$ 
\begin{align}
\overline z_{o,0}&=\N_o^\oplus \overline{x}_0 - \N_o^\ominus \underline{x}_0,\\
\underline z_{o,0}&=\N_o^\oplus \underline{x}_0 - \N_o^\ominus \overline{x}_0,
\end{align}
\end{subequations}
with $T \in \RR^{n_o \times n_o}$ solution to~\eqref{eq:sylvester}, for some $(A_o,B_o)$ picked as detailed above. Though observer~\eqref{eq:obszo} does not explicitly contain the chosen $A_o$, its dependence on $A_o$ is implicit in the resulting $T$, which will affect performance but is left for future studies. The following lemma can then be stated.

\begin{lemma}
   Suppose Assumption~\ref{ass_sys} holds. If $T$ solution to~\eqref{eq:sylvester} is invertible, observer~\eqref{eq:obszo} is an interval observer for system~\eqref{eq:syszo}.
\end{lemma}

\begin{pf}
    This is done by adapting the design in~\citep{baoThachAut} (the one written in the original coordinates). The initial conditions are because $Tz_{o,t} = T\N_ox_t$. The uncertainty bounds are due to the resulting disturbance $TD_od_t$ and noise $B_oWw_t$ after the transformation with $T$. The initial bounds follow from $z_{o,0} = \N_ox_0$.\hfill $\blacksquare$
\end{pf}

Note again that $T$ is almost always invertible---see Lemma~\ref{lem_T}. Therefore, our design does not have any additional conditions on system~\eqref{eq:sys} compared to traditional methods. Observer~\eqref{eq:obszo} is illustrated in the next example.
\begin{example}\label{eg2}
    Consider the system in Example~\ref{eg1}, after observability decomposition. For interval observer design, we assume the bounds $\overline{x}_0 = (1,1,1,1)$, $\underline{x}_0 = -\overline{x}_0$ for the initial condition, $\overline{d}_t = 0.02$, $\underline{d}_t = -0.02$ for the disturbance, and $\overline{w}_t = 0.01$, $\underline{w}_t = -0.01$ for the measurement noise. Our decomposition allows us to build for $z_{o,t}$ an interval observer~\eqref{eq:obszo} independently of $z_{no,t}$, regardless of whether $F_o$ is Metzler$\slash$non-negative or not (in our case, it is not). Recall that $(F_o,H_o)$ is observable with $n_y = 1$. By picking $A_o = \diag(0.1,0.2,0.3)$ (non-negative and Schur, no common eigenvalues with $F_o$) and $B_o = (1,1,1)$ (making $(A_o,B_o)$ controllable), we solve~\eqref{eq:sylvester} for an invertible matrix $T \in \RR^{3\times3}$. Observer~\eqref{eq:obszo} is then implemented, giving the results in Figure~\ref{fig:zo}. The state is initialized at $0$.
\end{example}
\begin{figure}[h]    \includegraphics[width=\columnwidth,height=0.34\columnwidth]{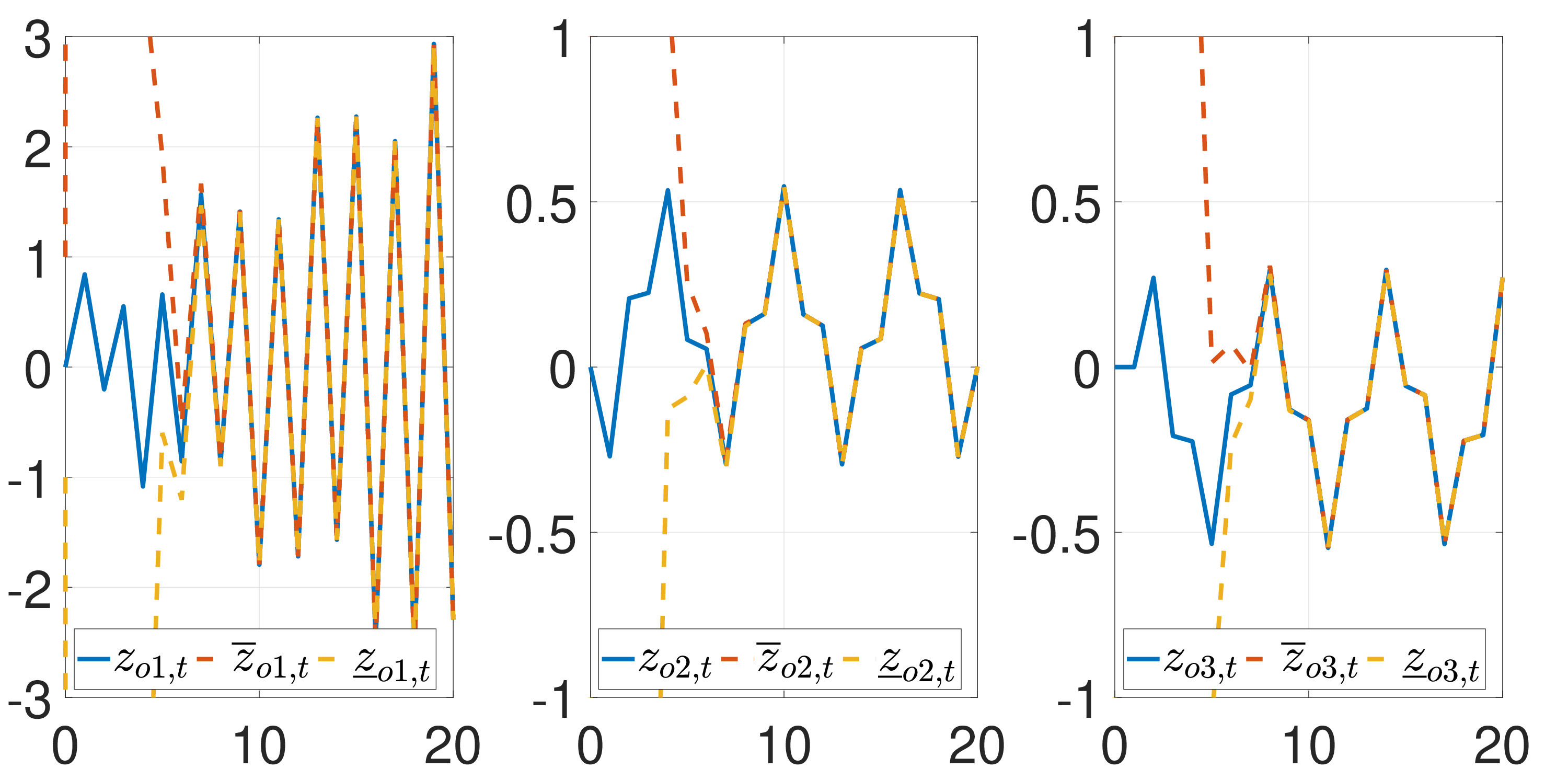}    \includegraphics[width=\columnwidth,height=0.34\columnwidth]{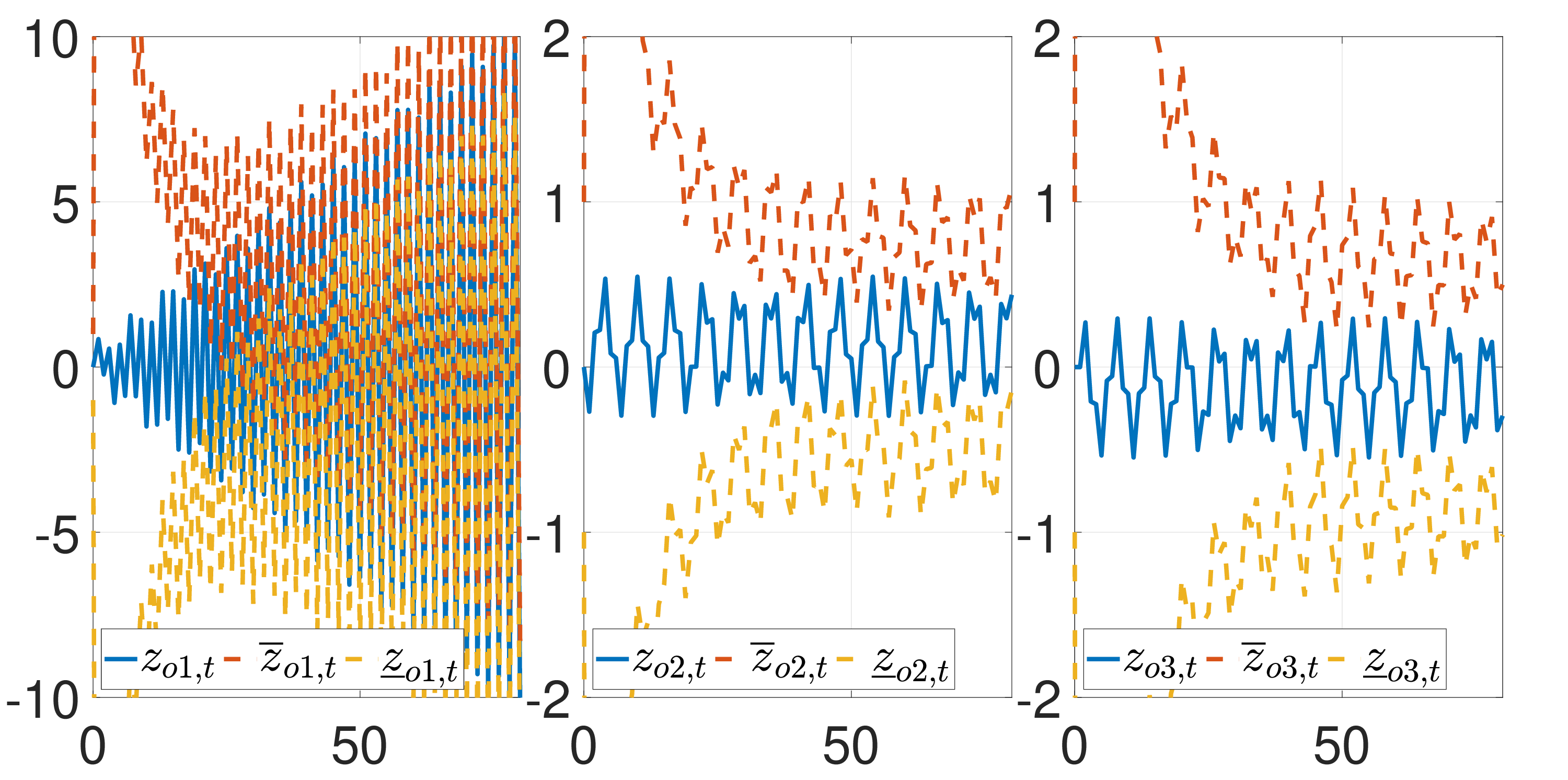}
        \caption{Interval estimation for $z_{o,t}$. Top: Convergence in the absence of $(d_t,w_t)$; bottom: Intervals in the presence of $(d_t,w_t)$.}
    \label{fig:zo}
\end{figure}

Now, we can guarantee interval in the $z_o$-coordinates. However, this is insufficient for recovering the bounds in the $x$-coordinates, because of $z_{no,t}$. We must then build another interval observer for $z_{no,t}$ in cascade with~\eqref{eq:obszo}. This is done in the next section.

\subsection{Jordan-Based Design for $z_{no,t}$}
In this section, we design an interval observer for $z_{no,t}$. The presence of $z_{o,t}$ in the dynamics of $z_{no,t}$ motivates us to see it as an exogenous signal with estimated bounds of
\begin{equation}\label{eq:syszno}
    z_{no,t}^+ = F_{noo}z_{o,t} + F_{no}z_{no,t} + \N_{no}u_t + D_{no}d_t,
\end{equation}
with $F_{no}$ stable, and no measured output. Therefore, when designing an interval observer for system~\eqref{eq:syszno}, we treat $F_{noo}z_{o,t}$ as another $D_{no}d_t$, with bounds provided by observer~\eqref{eq:obszo} that is independently guaranteed to work. In the CT case, following Lemma~\ref{lem_Pc}, given that $F_{no}$ is Hurwitz, there exist a time-varying norm-bounded matrix $P_t \in \RR^{n_{no} \times n_{no}}$ and a Metzler Hurwitz constant matrix $\Lambda\in \RR^{n_{no} \times n_{no}}$ such that $P_t^+ = \Lambda P_t - P_t F_{no}$. In the DT case, following Lemma~\ref{lem_Pd}, given that $F_{no}$ is Schur, there exist a sequence of norm-bounded matrices $(P_t)_{t \in \NN}$ in $\RR^{n_{no} \times n_{no}}$ and a non-negative Schur constant matrix $\Lambda\in \RR^{n_{no} \times n_{no}}$ such that $\Lambda = P_t^+ F_{no} (P_t)^{-1}$. The time-varying interval observer for system~\eqref{eq:syszno} has dynamics
\begin{subequations}\label{eq:obszno}
\begin{align}
\overline{\hat{z}}_{no,t}^+ & = \Lambda\overline{\hat{z}}_{no,t} + \Sigma_t\N_{no}u_t\notag\\&\qquad{}+ (\Sigma_t F_{noo})^\oplus \overline{z}_{o,t} - (\Sigma_t F_{noo})^\ominus \underline{z}_{o,t}\notag\\&\qquad{}+(\Sigma_t D_{no})^\oplus \overline{d}_t - (\Sigma_t D_{no})^\ominus \underline{d}_t,\\
\underline{\hat{z}}_{no,t}^+ & = \Lambda\underline{\hat{z}}_{no,t} + \Sigma_t\N_{no}u_t\notag\\&\qquad{}+ (\Sigma_t F_{noo})^\oplus \underline{z}_{o,t} - (\Sigma_t F_{noo})^\ominus \overline{z}_{o,t}\notag\\&\qquad{}+(\Sigma_t D_{no})^\oplus\underline{d}_t - (\Sigma_t D_{no})^\ominus \overline{d}_t,
\end{align}
with $\Sigma_t = P_t$ in CT and $\Sigma_t = P_t^+$ in DT, with $(\overline{z}_{o,t},\underline{z}_{o,t})$ from observer~\eqref{eq:obszo}, associated with the initial conditions
\begin{align}
\overline{\hat{z}}_{no,0} &= (P_0 \N_{no})^\oplus\overline{x}_0  - (P_0\N_{no})^\ominus\underline{x}_0,\\
\underline{\hat{z}}_{no,0} &= (P_0 \N_{no})^\oplus\underline{x}_0  - (P_0\N_{no})^\ominus\overline{x}_0,
\end{align}
with the bounds after time $0$
\begin{align}
\overline{z}_{no,t} &= (P_t^{-1})^\oplus \overline{\hat{z}}_{no,t} - (P_t^{-1})^\ominus \underline{\hat{z}}_{no,t},\\
\underline{z}_{no,t} &= (P_t^{-1})^\oplus \underline{\hat{z}}_{no,t} - (P_t^{-1})^\ominus \overline{\hat{z}}_{no,t},
\end{align}
and the bounds at time $0$ 
\begin{align}
\overline z_{no,0}&=\N_{no}^\oplus \overline{x}_0 - \N_{no}^\ominus \underline{x}_0,\\
\underline z_{no,0}&=\N_{no}^\oplus \underline{x}_0 - \N_{no}^\ominus \overline{x}_0,
\end{align}
\end{subequations}
with matrix $\Lambda$ and $t \mapsto P_t$ obtained as explained above. The following lemma can then be stated.
 
\begin{lemma}
   Suppose that Assumption~\ref{ass_sys} holds and that observer~\eqref{eq:obszo} is an interval observer for system~\eqref{eq:syszo}. Then, observer~\eqref{eq:obszno} is an interval observer for system~\eqref{eq:syszno}.
\end{lemma}

\begin{pf}
The observer structure~\eqref{eq:obszno} is similar to~\citep{Mazenccontinu,mazenc2014interval}, which differs between CT and DT by the computation of the transformation $P_t$. Since~\eqref{eq:obszo} is an interval observer for system~\eqref{eq:syszo}, we have $\underline{z}_{o,t} \leq z_{o,t} \leq \overline{z}_{o,t}$ for all $t \geq 0$ (in CT) or $t \in \NN$ (in DT), and so $z_{o,t}$ is seen as a disturbance alongside $d_t$, with its own known bounds. The initial conditions are obtained using Lemma~\ref{lem_pm} because we have $P_0 z_{no,0} = P_0 \N_{no}x_0$. We come back to the $z_{no}$-coordinates using Lemma~\ref{lem_pm} again. The initial bounds follow from $z_{no,0} = \N_{no}x_0$.\hfill $\blacksquare$
\end{pf}

Observer~\eqref{eq:obszno} is illustrated through the following example.
\begin{example}\label{eg3}
    Consider the system in Example~\ref{eg1}, for which we have designed a Sylvester-based interval observer for $z_{o,t}$ as in Example~\ref{eg2}. Our decomposition allows us to build for $z_{no,t}$ an interval observer~\eqref{eq:obszno}. Since $F_{no} = -0.5$ is Schur but is not non-negative, we can exploit the simple transformation $P_t = (-1)^t$. We obtain $\Lambda = P_t^+F_{no}(P_t)^{-1} = (-1)^{t+1}(-0.5)(-1)^t = 0.5$, which is non-negative and Schur. If we had had to compute and implement online this time-varying transformation for the whole $F \in \RR^{4\times 4}$, it would have been heavy. The initial and uncertainty bounds are in Example~\ref{eg2}. Observer~\eqref{eq:obszno} is then implemented, fed with the bounds from observer~\eqref{eq:obszo}, giving the results in Figure~\ref{fig:zno}. The state is initialized at $0$.
\end{example}
\begin{figure}[h]
\includegraphics[width=0.495\columnwidth,height=0.31\columnwidth]{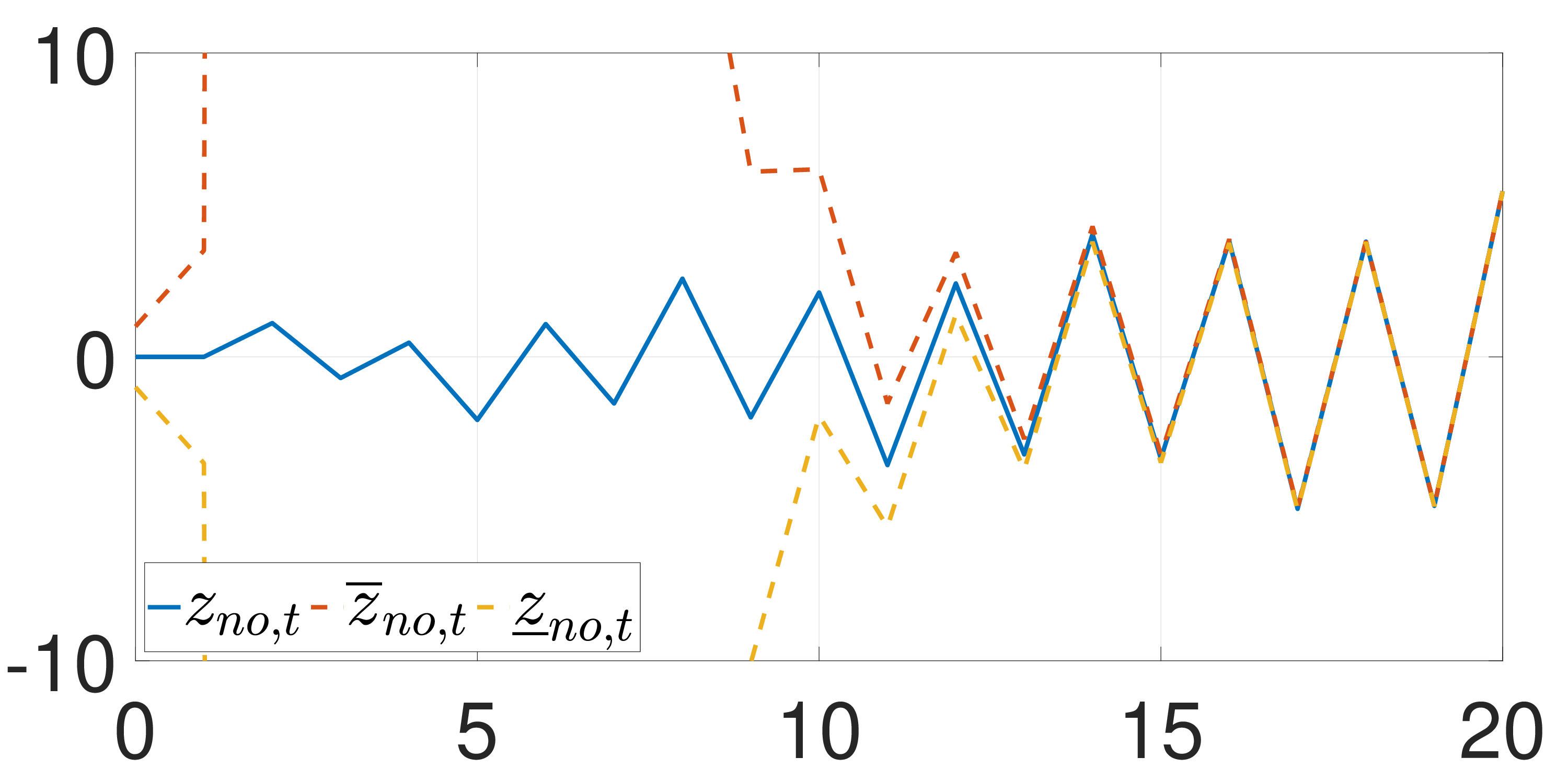}
\includegraphics[width=0.495\columnwidth,height=0.31\columnwidth]{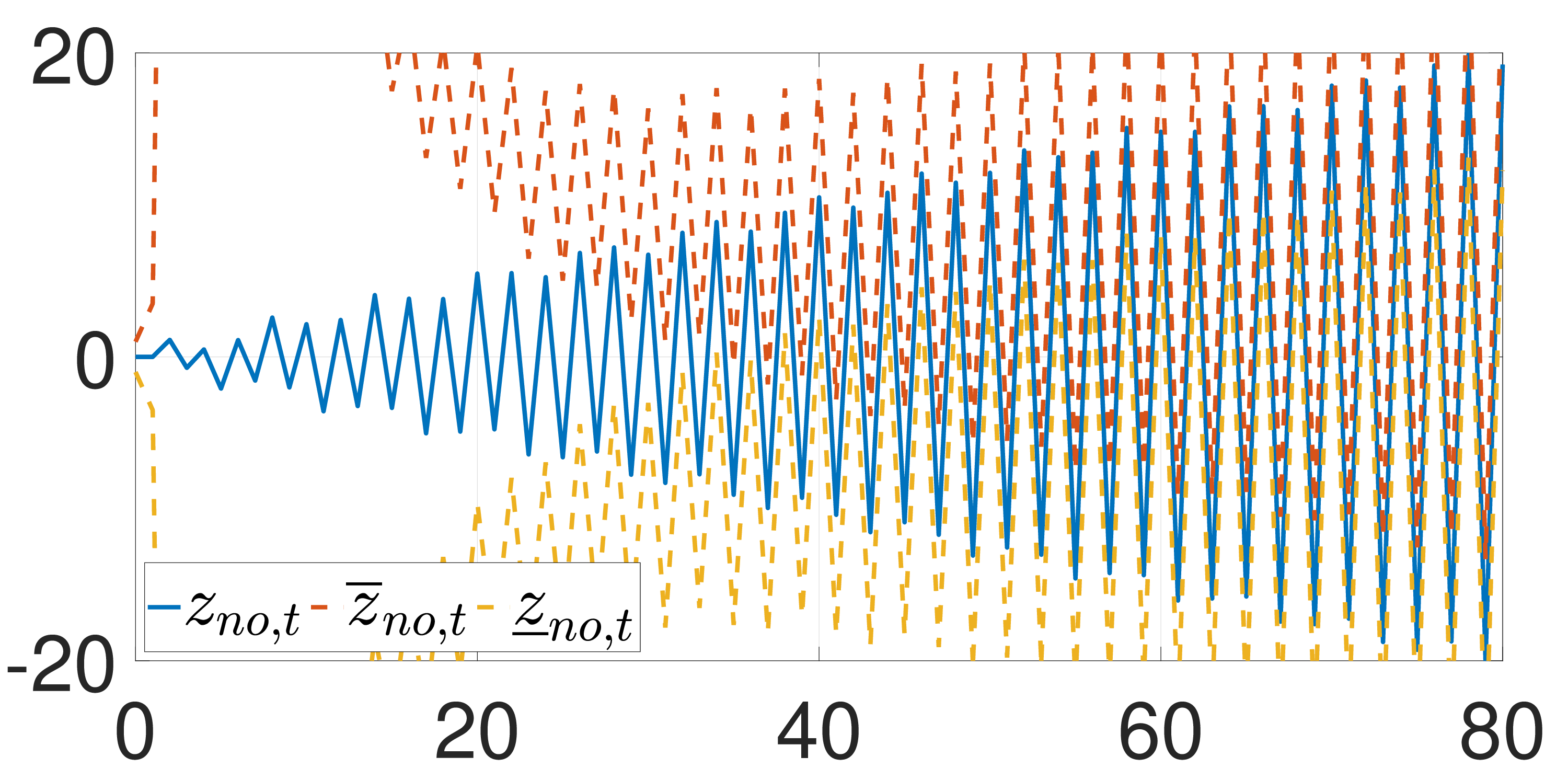}
        \caption{Interval estimation for $z_{no,t}$. Left: Convergence in the absence of $(d_t,w_t)$; right: Intervals in the presence of $(d_t,w_t)$.}
    \label{fig:zno}
\end{figure}

\subsection{Decomposition-Based Designs for $x_t$}
Now that we have built interval observers giving the bounds of $z_{o,t}$ and $z_{no,t}$, we use these to reconstruct the bounds of $x_t$. Naturally, the bounds in the $x$-coordinates are obtained after time $0$ as
\begin{subequations}\label{eq:boundx}
\begin{align}
\overline{x}_t & = \M_o^\oplus \overline{z}_{o,t} - \M_o^\ominus \underline{z}_{o,t} + \M_{no}^\oplus \overline{z}_{no,t} - \M_{no}^\ominus \underline{z}_{no,t}, \\
\underline{x}_t & = \M_o^\oplus \underline{z}_{o,t} - \M_o^\ominus \overline{z}_{o,t} + \M_{no}^\oplus \underline{z}_{no,t} - \M_{no}^\ominus \overline{z}_{no,t},
\end{align}
\end{subequations}
with $(\overline{z}_{o,t},\underline{z}_{o,t})$ from observer~\eqref{eq:obszo} and $(\overline{z}_{no,t},\underline{z}_{no,t})$ from observer~\eqref{eq:obszno}, and those at time $0$ are $(\overline{x}_0,\underline{x}_0)$ from Assumption~\ref{ass_sys}. These online bounds are straightforwardly obtained from the inverse~\eqref{eq:inv} using Lemma~\ref{lem_pm}.

Alternatively, thanks to linearity, we can invert the decomposition directly onto the observer dynamics, giving us a direct design in the $x$-coordinates taking the dynamics
\begin{mysubequations}\label{eq:obsx}
\begin{align}
\overline{\hat{x}}^+_t & =\begin{pmatrix}F_o-T^{-1}B_oH_o & 0 \\\Phi_t& \Lambda\end{pmatrix}\overline{\hat{x}}_t + \begin{pmatrix}0 & 0 \\\Omega_t& 0\end{pmatrix}\underline{\hat{x}}_t \notag\\&\qquad{}+\begin{pmatrix}T^{-1} (TD_o)^\oplus\\(\Sigma_t D_{no})^\oplus\end{pmatrix}\overline{d}_t -\begin{pmatrix}T^{-1} (TD_o)^\ominus\\(\Sigma_t D_{no})^\ominus\end{pmatrix}\underline{d}_t \notag\\&\qquad{}+\begin{pmatrix}T^{-1} (B_oW)^\ominus\\0\end{pmatrix}\overline{w}_t -\begin{pmatrix}T^{-1} (B_oW)^\oplus\\0\end{pmatrix}\underline{w}_t\notag\\&\qquad{}+\begin{pmatrix}\N_o \\ \Sigma_t \N_{no}\end{pmatrix} u_t +\begin{pmatrix}T^{-1}B_o\\0\end{pmatrix}y_t,\\
\underline{\hat{x}}^+_t& =\begin{pmatrix}F_o-T^{-1}B_oH_o & 0 \\\Phi_t& \Lambda\end{pmatrix}\underline{\hat{x}}_t + \begin{pmatrix}0 & 0 \\\Omega_t& 0\end{pmatrix}\overline{\hat{x}}_t \notag\\&\qquad{}+\begin{pmatrix}T^{-1} (TD_o)^\oplus\\(\Sigma_t D_{no})^\oplus\end{pmatrix}\underline{d}_t -\begin{pmatrix}T^{-1} (TD_o)^\ominus\\(\Sigma_t D_{no})^\ominus\end{pmatrix}\overline{d}_t \notag\\&\qquad{}+\begin{pmatrix}T^{-1}(B_oW)^\ominus\\0\end{pmatrix}\underline{w}_t -\begin{pmatrix}T^{-1} (B_oW)^\oplus\\0\end{pmatrix}\overline{w}_t\notag\\&\qquad{}+ \begin{pmatrix}\N_o \\ \Sigma_t \N_{no}\end{pmatrix} u_t +\begin{pmatrix}T^{-1}B_o\\0\end{pmatrix}y_t,
\end{align}
with $T$ solution to~\eqref{eq:sylvester} for appropriate $(A_o,B_o)$ as detailed above and $\Lambda$ obtained from $F_{no}$ using Lemma~\ref{lem_Pc} (in CT) or Lemma~\ref{lem_Pd} (in DT), where
\begin{align}
\Phi_t &= (\Sigma_tF_{noo})^\oplus (T^{-1})^\oplus T + (\Sigma_tF_{noo})^\ominus (T^{-1})^\ominus T,\\
\Omega_t &= -(\Sigma_tF_{noo})^\oplus (T^{-1})^\ominus T - (\Sigma_tF_{noo})^\ominus (T^{-1})^\oplus T,
\end{align}
with $\Sigma_t = P_t$ in CT and $\Sigma_t = P_t^+$ in DT, where $t \mapsto P_t$ is obtained from $F_{no}$ along with $\Lambda$ using Lemma~\ref{lem_Pc} (in CT) or Lemma~\ref{lem_Pd} (in DT), associated with the initial conditions
\begin{align}
\overline{\hat{x}}_0 & = \begin{pmatrix}T^{-1}(T\N_o)^\oplus\\(P_0 \N_{no})^\oplus\end{pmatrix}\overline{x}_0 - \begin{pmatrix}T^{-1}(T\N_o)^\ominus\\(P_0 \N_{no})^\ominus\end{pmatrix}\underline{x}_0,\\
\underline{\hat{x}}_0 & =\begin{pmatrix}T^{-1}(T\N_o)^\oplus\\(P_0 \N_{no})^\oplus\end{pmatrix}\underline{x}_0 - \begin{pmatrix}T^{-1}(T\N_o)^\ominus\\(P_0 \N_{no})^\ominus\end{pmatrix}\overline{x}_0,
\end{align}
and the bounds after time $0$
\begin{align}
\overline{x}_t & = \begin{pmatrix}\phi_l&\varphi_{l,t} \end{pmatrix}\overline{\hat{x}}_t - \begin{pmatrix}\phi_r&\varphi_{r,t} \end{pmatrix}\underline{\hat{x}}_t,\\
\underline{x}_t & =\begin{pmatrix}\phi_l&\varphi_{l,t} \end{pmatrix}\underline{\hat{x}}_t - \begin{pmatrix}\phi_r&\varphi_{r,t} \end{pmatrix}\overline{\hat{x}}_t,
\end{align}
where
\begin{align}
\phi_l & = \M_o^\oplus (T^{-1})^\oplus T + \M_o^\ominus(T^{-1})^\ominus T,\\
\varphi_{l,t}& = \M_{no}^\oplus (P_t^{-1})^\oplus + \M_{no}^\ominus(P_t^{-1})^\ominus,\\
\phi_r & = \M_o^\oplus (T^{-1})^\ominus T + \M_o^\ominus(T^{-1})^\oplus T,\\
\varphi_{r,t}& = \M_{no}^\oplus (P_t^{-1})^\ominus + \M_{no}^\ominus(P_t^{-1})^\oplus.
\end{align}
\end{mysubequations}
The bounds at time $0$ are from Assumption~\ref{ass_sys}.

Observer~\eqref{eq:obsx} is obtained by concatenating $(\overline{\hat{z}}_{o,t},\overline{\hat{z}}_{no,t})$ and $(\underline{\hat{z}}_{o,t},\underline{\hat{z}}_{no,t})$ and denoting them as the new $x$-coordinates states $\overline{\hat{x}}_t$ and $\underline{\hat{x}}_t$. Even though the forms look complicated, a considerable part of these are in fact \emph{constant} (namely those without the subscript $t$) and can be computed offline, reducing computational burden. The larger $n_o$, namely the larger the observable part of the state with respect to the full state, the larger the constant part in observer~\eqref{eq:obsx}. The results in the $x$-coordinates are stated in the next lemma.

\begin{lemma}
Suppose Assumption~\ref{ass_sys} holds and $T$ solution to~\eqref{eq:sylvester} is invertible. Then, observers~\eqref{eq:obszo}-\eqref{eq:obszno}-\eqref{eq:boundx} and~\eqref{eq:obsx} are interval observers for system~\eqref{eq:sys}.
\end{lemma}

The step-by-step procedure to design our decomposition-based interval observers is presented in Algorithm~\ref{alg:cap}.

\begin{algorithm}
\caption{Decomposition-based interval observer design for system~\eqref{eq:sys}}\label{alg:cap}
\begin{algorithmic}
\Require System~\eqref{eq:sys} under Assumption~\ref{ass_sys}
\State \textbf{|Offline steps|}
\State \textbf{Step 1:} Analyze observability with $\O$
\State \textbf{Step 2:} Find matrices $\M_o$, $\M_{no}$, $\N_o$, $\N_{no}$
\State \textbf{Step 3:} Compute matrices of system~\eqref{eq:sysz} 
\State \textbf{Step 4:} Pick $(A_o,B_o)$ and solve~\eqref{eq:sylvester} for $T$ \Comment{In Matlab, use {\fontfamily{qcr}\selectfont T = sylvester(Ao,-Fo,-Bo*Ho)}} 
\State \textbf{Step 5:} Find $\Lambda$ and $t \mapsto P_t$ from $F_{no}$ using~\citep[Theorem $2$]{Mazenccontinu} (in CT) or~\citep[Theorem $4$]{mazenc2014interval} (in DT) 
\State \textbf{|Online steps|} \State \textbf{Step 6:} Implement observer~\eqref{eq:obszo}-\eqref{eq:obszno}-\eqref{eq:boundx} or~\eqref{eq:obsx}
\end{algorithmic}
\end{algorithm}

\begin{example}\label{eg4}
    Consider the system in Example~\ref{eg1}, for which we have designed Sylvester-based and Jordan-based interval observers for $z_{o,t}$ and $z_{no,t}$ as in Examples~\ref{eg2} and~\ref{eg3}, respectively. Observer~\eqref{eq:obszo}-\eqref{eq:obszno}-\eqref{eq:boundx} is then implemented by simply running~\eqref{eq:boundx} on the $z_o$- and $z_{no}$-intervals, giving the results in Figure~\ref{fig:x}. We could alternatively implement observer~\eqref{eq:obsx}, which would give similar results. Note again that finding and implementing online a $4 \times 4$ time-varying transformation for the system in Example~\ref{eg1} using~\citep{mazenc2014interval} is rather computationally heavy, illustrating the interest of our decomposition-based approach. 
\end{example}
\begin{figure}[h]    \includegraphics[width=\columnwidth,height=0.5\columnwidth]{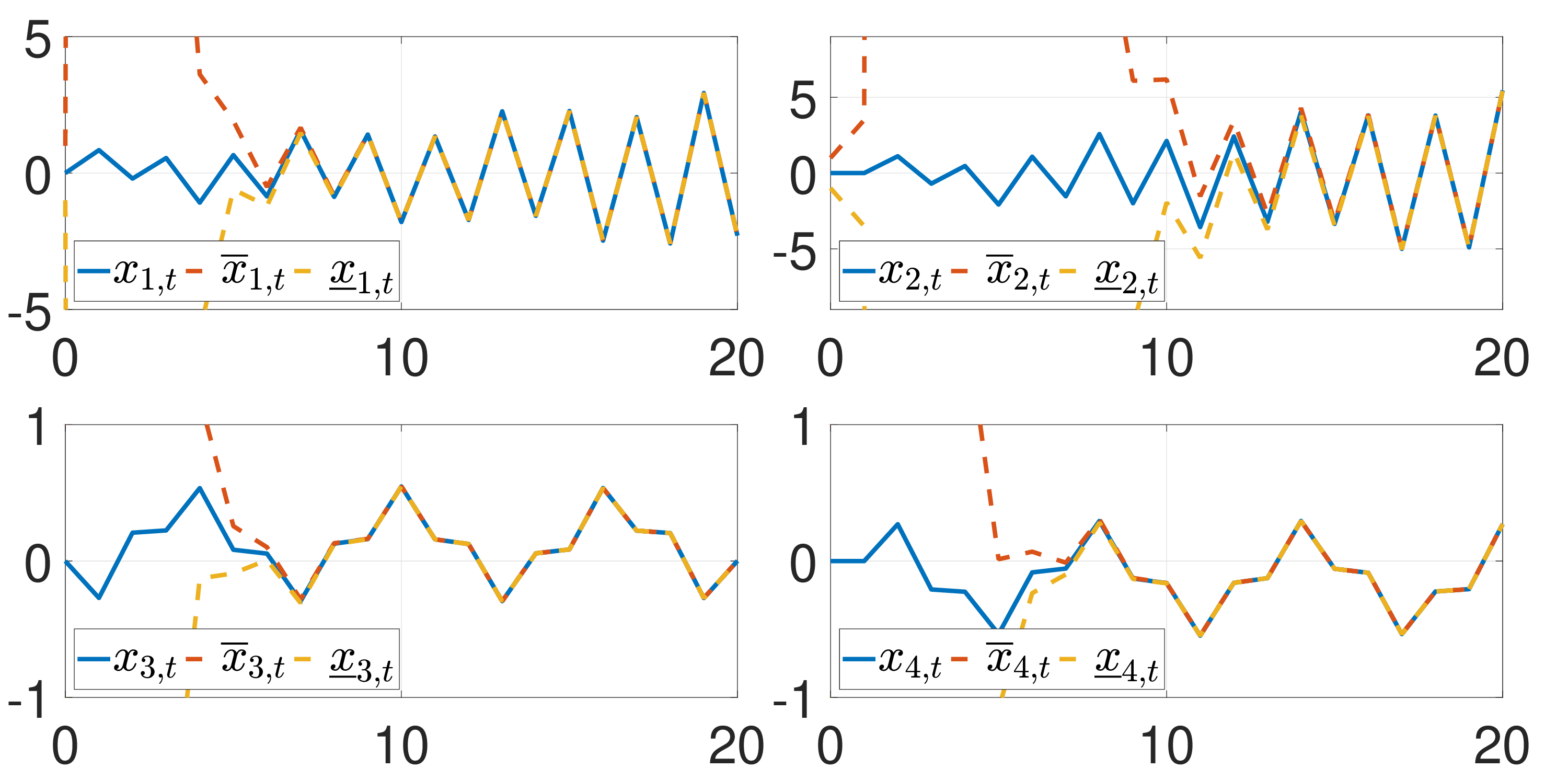}
\includegraphics[width=\columnwidth,height=0.5\columnwidth]{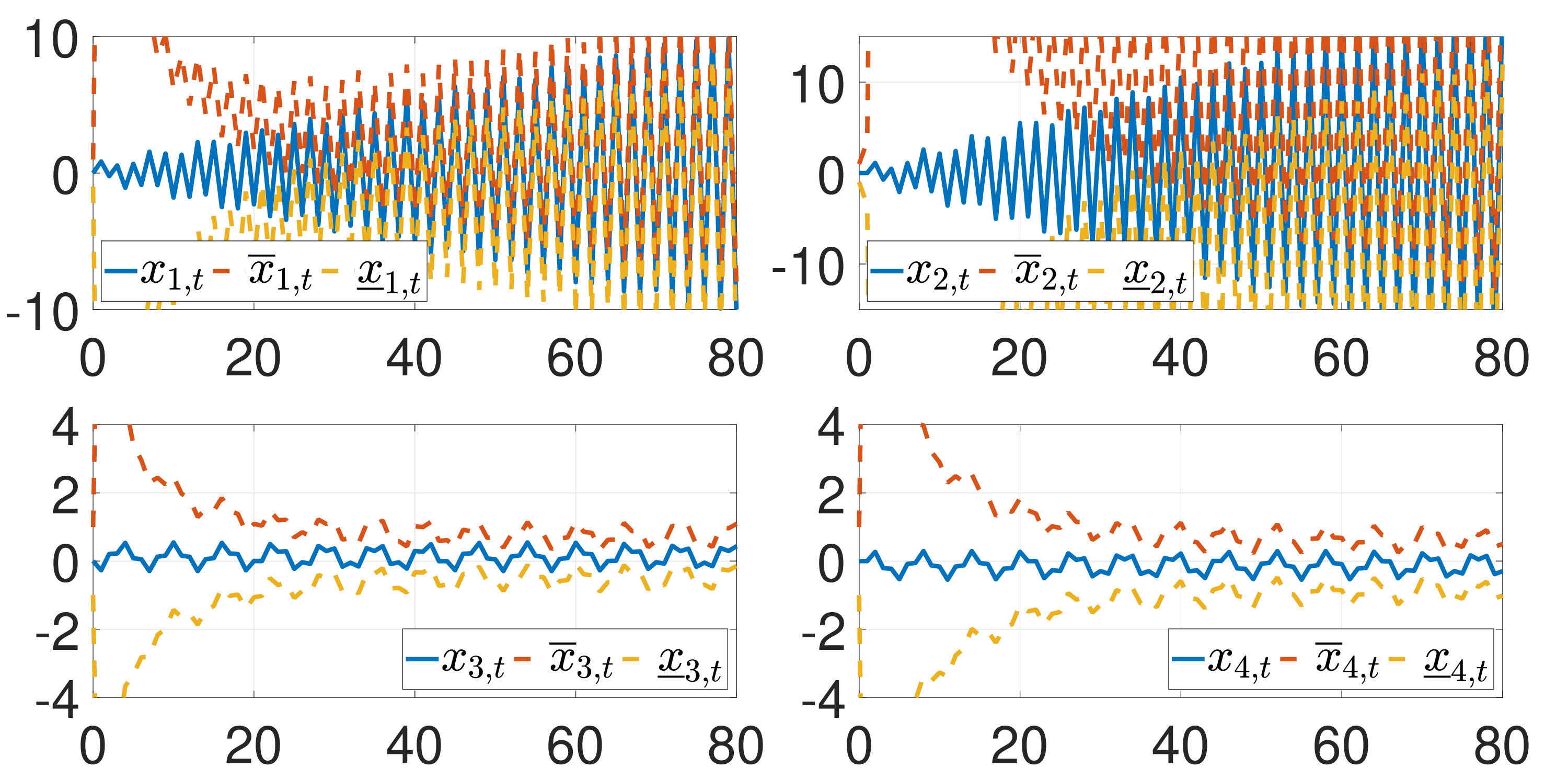}
        \caption{Interval estimation for $x_t$ using observer~\eqref{eq:obszo}-\eqref{eq:obszno}-\eqref{eq:boundx}. Top: Convergence in the absence of $(d_t,w_t)$; bottom: Intervals in the presence of $(d_t,w_t)$.}
    \label{fig:x}
\end{figure}

\section{CONCLUSION}
We propose a systematic interval observer for detectable LTI systems, where a decomposition separates the state into two parts---one observable from the output and one detectable---for which suitable interval observers are designed and concatenated. We recover the bounds either by inverting the decomposition on the new bounds or by directly writing the combined observer in the original coordinates, resulting in reduced computational complexity.

Note that since the results in~\citep{baoThachAut} extend to LTV systems, the approach naturally applies to time-varying $(F_o,H_o)$ under stronger observability conditions and possibly higher-dimensional $z_o$. This work thus covers the class of LTV$\slash$LPV systems admitting a decomposition of the form~\eqref{eq:sysz} resulting in a constant $F_{no}$ as required by Lemmas~\ref{lem_Pc} and~\ref{lem_Pd}, while the other matrices can be varying.

Future work includes optimizing $T$ and parameters $(A_o,B_o)$ satisfying~\eqref{eq:sylvester} to tighten interval observer bounds, both in this paper and~\citep{baoThachAut}, using norm-based optimization methods, e.g.,~\citep{li19}.

\begin{ack} 
The work of G. Q. B. Tran was supported in part by the AFOSR MURI FA9550-23-1-0337 grant. The work of T. N. Dinh was supported in part by the Hanse-Wissenschaftskolleg---Institute for Advanced Study (HWK).\end{ack}

\bibliography{ref} 
\end{document}